\documentclass[12pt]{article}

\usepackage[english]{babel}
\usepackage[letterpaper,top=1in,bottom=1in,left=1in,right=1in]{geometry}
\usepackage{amsmath,amssymb,amsthm}
\usepackage{graphicx}
\usepackage{booktabs}
\usepackage{multirow}
\usepackage{makecell}
\usepackage{setspace}
\usepackage{enumitem}

\newtheorem{theorem}{Theorem}
\newtheorem{corollary}{Corollary}[theorem]
\usepackage[numbers]{natbib}
\usepackage[colorlinks=true,allcolors=blue]{hyperref}

\title{An Approximate Maximum Likelihood Estimator for Discretely Observed Linear Birth-and-Death Processes}
\author{
    Xiaochen Long\thanks{\texttt{xiaochen.long@rice.edu}} \\
    Department of Statistics, Rice University
    \and
    Marek Kimmel\thanks{\texttt{kimmel@rice.edu}} \\
    Department of Statistics, Rice University
}
\begin{document}

\maketitle

\begin{abstract}
Linear birth-and-death processes (LBDPs) are foundational stochastic models in population dynamics, evolutionary biology, and hematopoiesis. Estimating parameters from discretely observed data is computationally demanding due to irregular sampling, noise, and missing values. We propose a novel approximate maximum likelihood estimator (MLE) for LBDPs based on a Gaussian approximation to transition probabilities. The approach transforms estimation into a univariate optimization problem, achieving substantial computational gains without sacrificing accuracy.

Through simulations, we show that the approximate MLE outperforms Gaussian and saddlepoint-based estimators in speed and precision under realistic noise and sparsity. Applied to longitudinal clonal hematopoiesis data, the method produces biologically meaningful growth estimates even with noisy, compositional input. Unlike Gaussian and saddlepoint approximations, our estimator is invariant to data scaling, making it ideal for real-world applications such as variant allele frequency analyses.
\end{abstract}


\newpage

\section{Introduction}

Linear birth-death processes (LBDPs) are a fundamental class of continuous-time Markov processes widely used in population genetics and other domains. Owing to their time-homogeneous and memoryless (Markov) properties, LBDPs are particularly suitable for modeling populations where individuals are considered equivalent and exchangeable. This modeling assumption is common in scenarios such as cell expansion, as demonstrated in several recent studies \cite{watson2020evolutionary,biesiadny2022statistical,fabre2022longitudinal,zarebski2022computationally,louca2020extant,wu2023using}, as well as in epidemiology and related fields.

The challenge of estimating parameters from discretely observed LBDPs was recognized as early as 1951 \cite{immel1951problems}, and remains pertinent in many contemporary applications. Although various estimation approaches have been proposed, robust and efficient estimation remains difficult in practice. Challenges include handling noisy or zero-inflated data due to missing or undetected observations, dealing with datasets that contain limited temporal observations per individual but involve large sample sizes across individuals, and scaling estimation procedures to more complex hierarchical models based on LBDPs. These scenarios demand methods that are both computationally efficient and statistically robust to data irregularities.

In this paper, we propose a fast and accurate estimation method for LBDPs based on the Gaussian approximation. By reducing the estimation task to a univariate optimization problem, our method achieves a balance between accuracy and computational efficiency, outperforming the Gaussian and saddlepoint approximations in runtime while maintaining high estimation accuracy. We also demonstrate that this method extends naturally to time-inhomogeneous birth-death processes, where the birth and death rates, $\lambda(t)$ and $\mu(t)$, vary over time.

We evaluate the performance of four estimation approaches under non-equidistant observation scenarios: the Galton-Watson estimator (as a misspecified baseline model), Gaussian approximation, saddlepoint approximation, and our proposed method. Our simulation results highlight the stability and precision of the proposed estimator, which also surpasses the other methods in computational speed. We further apply the method to a real-world dataset (e.g., the ARIC study) to estimate the growth rates of clonal hematopoiesis of indeterminate potential (CHIP), showcasing its effectiveness in the presence of noise and missing data.

The supplementary material provides theoretical derivations, convergence guarantees for the proposed estimator, an interpretation of the approach as a weighted regression problem, and a simplified form of the LBDP transition probability to facilitate faster computation.

\section{Linear Birth-Death Processes (LBDP)}

\subsection{Definition}

The history of the linear birth-death process can be dated back to \cite{kendall1948generalized}, where analytical results of the time-nonhomogeneous birth-death process were extensively given. Here we use the definition of the time-nonhomogeneous linear birth-death process under Kendall's statement: let the integer-valued time-dependent random variable $X(t)$ measure at time $t$ the size of a population (cell population, particles, etc.) and suppose that in an element of time $dt$ the only possible transitions (and their associated probabilities) are:

\begin{equation}
\begin{aligned}
    X(t+dt) &= X(t) + 1,  & \text{with associated probability } \lambda(t)X(t)\, dt + o(dt);\\
    X(t+dt) &= X(t),      & \text{with associated probability } 1 - \left[\lambda(t) + \mu(t)\right]X(t)\, dt + o(dt);\\
    X(t+dt) &= X(t) - 1,  & \text{with associated probability } \mu(t)X(t)\, dt + o(dt)
\end{aligned}
\end{equation}

Note that in LBDP, the birth rate $\lambda(t)$ and death rate $\mu(t)$ are assumed to be constant, denoted by $\lambda$ and $\mu$.

\subsection{Estimation of LBDPs}

The estimation of linear birth-death processes (LBDPs), particularly for the birth rate $\lambda$ and death rate $\mu$, has been extensively studied in the literature. \cite{keiding1975maximum} and \cite{guttorp1991statistical} provide exact analytical results for maximum likelihood estimation (MLE) under the assumption of equidistant time points and continuous-time observations. \cite{wu2023using} explores the use of Gaussian approximations for estimating LBDPs in population genetics. \cite{davison2021parameter} introduces a saddlepoint approximation approach that allows accurate estimation without relying solely on the exact transition probabilities. Additionally, \cite{tavare2018linear} discusses Bayesian and Approximate Bayesian Computation (ABC) techniques for estimating LBDPs, also covering various frequentist methods applicable to general birth-death processes. Some of these techniques, including those in \cite{crawford2018computational}, are adaptable for use in the context of LBDPs.

Despite these advancements, several challenges remain in applying LBDP models effectively:

\begin{itemize}
    \item \textbf{Zero observations within sequences:} Direct optimization methods for LBDPs cannot accommodate zero counts occurring between nonzero observations, as a zero indicates extinction, violating model assumptions.
    \item \textbf{Numerical instability:} Estimations may become unstable when handling extreme values or noisy data, especially in human experiments where undetected or missing data are often recorded as zeros, and sample variability can be high.
    \item \textbf{Computational inefficiency:} The computational complexity of many LBDP estimation techniques limits their applicability in hierarchical or large-scale models.
\end{itemize}

Consequently, many LBDP estimation methods become infeasible for datasets characterized by noise or limited individual-level observations.






\section{Discretely-Observed LBDP}
\subsection{Galton-Watson Estimator}

As our observation timepoints are discrete and can be non-equidistant, it is also called discretely-observed LBDP. If the time interval between each observation is equal as $\tau$, it is called "the discrete skeleton" as it can form a Galton-Watson process in discrete time. \cite{keiding1975maximum} uses a smart way based on the skeleton of the Galton-Watson process to derive the transition probability $P_{nm}(t) = P\left[X(T+t)=m\vert X(T)=n\right], \ \forall \ T\geq 0$ of this case:

\begin{equation}
\label{eq:Keiding transition}
    P_{nm}(t) = \sum_{j=0}^{min(m,n)}\binom{n}{j}\binom{n+m-j-1}{n-1} A(t)^{n-j}B(t)^{m-j}\left[1-A(t)-B(t)\right]^{j}
\end{equation}

where 

\begin{equation}
    A(t) = \frac{\mu\left[e^{(\lambda-\mu)t}-1\right]}{\lambda e^{(\lambda-\mu)t}-\mu},\quad B(t) = \frac{\lambda}{\mu} A(t)
\end{equation}

During our study, we discovered an alternative expression for the transition probability of the LBDP, which is equivalent to \eqref{eq:Keiding transition}:

\begin{equation}
\label{eq:LBDP transition}
    P_{nm}(t) = \sum_{j=1}^{\min(n,m)} \binom{m-1}{j-1} \binom{n}{j} A(t)^{n-j} B(t)^{m-j} \left[1-A(t)\right]^j \left[1-B(t)\right]^j.
\end{equation}

Alternatively, this transition probability can be expressed using the hypergeometric function ${}_2F_1$ as follows:

\begin{equation}
\label{eq:LBDP transition 2F1}
    P_{nm}(t) = A(t)^{n}B(t)^{m} \binom{n+m-1}{m} {}_2F_1\left(-m, -n; 1-n-m; \frac{1}{A(t)}+\frac{1}{B(t)}-1\right).
\end{equation}

These formulas, \eqref{eq:LBDP transition} and \eqref{eq:LBDP transition 2F1}, are derived by leveraging the property that the sum of independent and identically distributed (iid) geometric random variables follows a negative binomial distribution. Specifically, conditioned on survival, the progeny distribution of the LBDP is a geometric distribution starting at one. By summing up to \( j \), we obtain the probability \( P_{nm}(t) \) for a surviving family of \( j \) ancestors at time \( 0 \). (See Appendix for further details.)

In equidistant timepoints scenario, the estimation based on \eqref{eq:Keiding transition} or \eqref{eq:LBDP transition} are easy as $A(t)$ and $B(t)$ are constants. \cite{keiding1975maximum} gives a simple but direct derivation of the MLE of growth rate $\alpha=\lambda-\mu$ in this case:

\begin{align} \label{eq:Galton-Watson estimator}
\widehat{\alpha} = \frac{1}{\tau} \log \left[\frac{X(\tau) + \dots + X({k\tau})}{X(0) +\dots + X\left[(k-1)\tau\right]}\right]
\end{align}

This estimator is frequently called the "Galton-Watson estimator" as the transition probability is homogeneous, and LBDP can be proved to be identical to the linear fractional GW process at equidistant discrete time points. In this paper, we will refer to this estimator as the Galton-Watson estimator.

A simple MLE in analytical form is desirable, but not universally applicable. In most experiments, "equidistant timepoints" means all patients have to be tested and observed at the same time intervals, which is not very common, especially for human studies (such as the ARIC study). Unfortunately, when the observation times are not equidistant, the MLE does not have an elegant analytical form or can be derived analytically. It has to be computed numerically. Especially when $X(t)$ is large, the transition probability becomes intractable (as calculation includes combination of $X(t)$), we can only approximate the transition probability by Gaussian approximation \citep{wu2023using} or saddplepoint approximation \citep{davison2021parameter}. However, they still have some limitations.

When applying Gaussian approximation or saddlepoint approximation on real-world data, we noticed that the optimizer does not always converge, while the Galton-Watson estimator can give a fast and reasonable estimate. This is because in real-world data, extra noises and observation errors are introduced in the data, so that the method, which is highly dependent on LBDP likelihoods and its functional form, will not be feasible. 

However, the Galton-Watson estimator, although derived based on likelihood maximization, is pretty robust against noises due to its simplicity. Moreover, the Galton-Watson estimator only needs to maximize over one parameter: $\alpha$, while the Gaussian approximation and saddlepoint approximation need to maximize over $\lambda$ and $\mu$ jointly, or equivalently in terms of the Malthusian parameter $\alpha=\lambda-\mu$ and $\beta=\lambda/\mu$, which is harder, especially for the ratio parameter $\beta$. 

\section{An Approximate MLE for LBDP}
\subsection{Derivation}
Our intuition is to start from a Gaussian approximation, and try to approximate the likelihood to obtain an estimate of $\alpha$ based on a univariate optimization like the Galton-Watson Estimator.

In the Gaussian approximation of LBDP, we consider the transition distribution:

\begin{align} \label{eq:Gaussian_approximation_2}
X(t) \vert X(0) \sim N\left( X(0) e^{(\lambda-\mu)t}, \ \frac{\lambda+\mu}{\lambda-\mu}X(0) e^{(\lambda-\mu)t} \left[e^{(\lambda-\mu)t}-1\right] \right),
\end{align}

which leads to the probability density function:

\begin{align} \label{eq:Gaussian Approximation Likelihood}
    f_{X(t) \mid X(0)}(x) = \frac{1}{\sqrt{2\pi \sigma_X^2(t)}} 
    \exp\left( -\frac{(x - \mu_X(t))^2}{2\sigma_X^2(t)} \right),
\end{align}

where the mean and variance are defined as:

\begin{align}
\begin{split}
    \mu_X(t) &= X(0) e^{(\lambda-\mu)t}, \\
    \sigma_X^2(t) &= \frac{\lambda+\mu}{\lambda-\mu} X(0) e^{(\lambda-\mu)t} \left[e^{(\lambda-\mu)t}-1\right].\\
\end{split}
\end{align}

To simplify notation and facilitate further derivations, we define parameters to replace $\lambda,\mu$:

\begin{equation}
    \alpha = \lambda - \mu, \quad \sigma^2 = \frac{\lambda+\mu}{\lambda-\mu}.
\end{equation}

\subsubsection{Case 1: Equidistant Time Intervals}
\label{subsubsec:Galton-Watson}
When \(t_1=t_2= \dots =t_{n-1}\), defining \(t^* = \exp(\alpha t_1)\), the MLE of the Gaussian approximation is identical to the true MLE of LBDP derived by \cite{keiding1975maximum}:

\begin{align}
    &\hat{t^*} = \frac{X_2+X_3+\dots+X_n}{X_1+X_2+\dots+X_{n-1}}, \\
    &\hat{\alpha} = \frac{1}{t_1} \log \frac{X_2+X_3+\dots+X_n}{X_1+X_2+\dots+X_{n-1}}, \quad\quad (\text{Galton-Watson Estimator})\\
    &\hat{\sigma}^2 =  \sum_{i=1}^{n-1} \left\{\sqrt{\frac{X_i\exp(\hat{\alpha} t_i)}{\exp(\hat{\alpha} t_i)-1}} \left[\frac{X_{i+1}}{\exp(\hat{\alpha} t_i)} - X_i \right] \right\}^2 / (n-1).
\end{align}

\subsubsection{Case 2: Non-Equidistant Time Intervals}

For general non-equidistant time points, the transition distribution of \(X_{i+1} \mid X_i\) is:

\begin{align}
X_{i+1} \mid X_i \sim N\left(X_i\exp(\alpha t_i), X_i \exp(\alpha t_i)\left[\exp(\alpha t_i)-1\right] \sigma^2\right).
\end{align}

Thus, the likelihood function of \(X_{i+1} \mid X_i\) is given by:

\begin{align}
L(\alpha, \sigma^2) = \frac{1}{\sigma \sqrt{2\pi X_i\exp(\alpha t_i)\bigl[\exp(\alpha t_i)-1\bigr]}} \exp\left\{-\frac{[X_{i+1} - X_i\exp(\alpha t_i)]^2 }{2X_i\exp(\alpha t_i)\bigl[\exp(\alpha t_i)-1\bigr]\sigma^2}\right\}.
\end{align}

The partial derivative of the log-likelihood function $\frac{\partial l(X_{i+1} \mid X_i)}{\partial \alpha} $ is:

\begin{align}
\frac{\partial l}{\partial \alpha} =  -\frac{t_i}{2} - \frac{t_i\exp(\alpha t_i)}{2[\exp(\alpha t_i)-1]} -\frac{\partial }{\partial \alpha} \frac{[X_{i+1} - X_i\exp(\alpha t_i)]^2 }{2X_i\sigma^2\exp(\alpha t_i)[\exp(\alpha t_i)-1]}.
\end{align}

which can be simplified to 

\begin{align}
\label{eq:likelihood decomposition 1}
\frac{\partial l}{\partial \alpha} &= \frac{t_i\Bigl[X_{i+1} - X_i\exp(\alpha t_i)\Bigr]}{\sigma^2\Bigl[\exp(\alpha t_i)-1\Bigr]}  \\
&\quad + \frac{t_i\Bigl[X_{i+1} - X_i\exp(\alpha t_i)\Bigr]^2\Bigl[2\exp(\alpha t_i)-1\Bigr]}{2X_i\sigma^2\exp(\alpha t_i)\Bigl[\exp(\alpha t_i)-1\Bigr]^2}  \\
&\quad -\frac{t_i}{2} - \frac{t_i\exp(\alpha t_i)}{2[\exp(\alpha t_i)-1]}.
\end{align}

We can approximate the likelihood function of $X_{i+1}\vert X_i$ by only the first term \eqref{eq:likelihood decomposition 1} because it will dominate the partial derivative as $X_i\rightarrow\infty$.

\begin{theorem}
\label{theorem:Dominating term}
Define $l_1, l_2, l_3$ as the three additive terms in \eqref{eq:likelihood decomposition 1} as:

\begin{align}
\begin{split}
\label{Def: three terms}
& l_1 = \frac{t_i\Bigl[X_{i+1} - X_i\exp(\alpha t_i)\Bigr]}{\sigma^2\Bigl[\exp(\alpha t_i)-1\Bigr]} \\
& l_2 = \frac{t_i\Bigl[X_{i+1} - X_i\exp(\alpha t_i)\Bigr]^2\Bigl[2\exp(\alpha t_i)-1\Bigr]}{2X_i\sigma^2\exp(\alpha t_i)\Bigl[\exp(\alpha t_i)-1\Bigr]^2}\\
& l_3 = -\frac{t_i}{2} - \frac{t_i\exp(\alpha t_i)}{2[\exp(\alpha t_i)-1]}\\
\end{split}
\end{align}.

As $X_i\rightarrow\infty$, $l_1$ will dominate \eqref{eq:likelihood decomposition 1} as $\frac{l_2}{l_1}\rightarrow 0$ and $\frac{l_3}{l_1}\rightarrow 0$.

\end{theorem}

This dominance result is relevant when $X_i$ is large, a common assumption in Gaussian or saddlepoint approximations and often met in biological experiments with large population sizes. Here, $X_i \to \infty$ means the approximation holds for large observed values even with finitely many observations, and is independent of any constraint on $\alpha$.

For computational efficiency, we ignore the second and third terms $l_2 + l_3$, as they are asymptotically negligible when $X_i$ is large by Theorem~\ref{theorem:Dominating term}, leading to the reduced optimization equation:

\begin{align}
\label{eq:partial derivative approximation}
\frac{\partial l}{\partial \alpha} \approx \frac{1}{\sigma^2}\sum_{i=1}^{n-1} \frac{t_i\Bigl[X_{i+1} - X_i\exp(\alpha t_i)\Bigr]}{\exp(\alpha t_i)-1}.
\end{align}

Setting this derivative to zero yields the approximate estimator:

\begin{align}
\label{eq:approximate estimator equation}
h(\alpha) = \sum_{i=1}^{n-1} \frac{t_i}{\exp(\alpha t_i)-1} \Bigl[X_{i+1} - X_i\exp(\alpha t_i)\Bigr] =0.
\end{align}

while the other parameter $\sigma^2=\frac{\lambda+\mu}{\lambda-\mu}$'s estimate can be computed using this formula directly with $\hat{\alpha}$ as the solution of \eqref{eq:approximate estimator equation}:

\begin{equation}
\hat{\sigma}^2 =  \sum_{i=1}^{n-1} \left\{\sqrt{\frac{X_i\exp(\hat{\alpha} t_i)}{\exp(\hat{\alpha} t_i)-1}} \left[\frac{X_{i+1}}{\exp(\hat{\alpha} t_i)} - X_i \right] \right\}^2 / (n-1).
\end{equation}

\subsection{Properties}
\subsubsection{Convergence of the approximate estimator}
Define the true parameter $\alpha$ as $\alpha_0$, we can show that this estimator $\hat{\alpha}\rightarrow\alpha_0$ as $X_i\rightarrow \infty$.

\begin{theorem}[Convergence of the estimator equation]
\label{theorem:convergence of estimator equation}
Solving \eqref{eq:approximate estimator equation} is equivalent to solving
\begin{equation}
\label{eq:consistency of function}
g(\alpha) =\frac{1}{X_1\sum_{i=1}^{n-1} t_i\exp(\alpha_0 T_i)} \left[\sum_{i=1}^{n-1}  \frac{t_i}{\exp(\alpha t_i)-1}(X_{i+1}-X_i) - \sum_{i=1}^{n-1} t_i X_i\right] = 0.
\end{equation}

Suppose $X_{\min} = \min(X_1, \dots, X_n)$. As $X_{\min} \rightarrow \infty$, the function $g(\alpha)$ converges to a limiting function $g^*(\alpha)$ given by:
    \begin{equation}
    \label{eq:limit of g(alpha)}
        \lim_{X_{\min} \to \infty} g(\alpha) = g^*(\alpha) = \frac{1}{X_1 \sum_{i=1}^{n-1} t_i \exp(\alpha_0 T_i)} 
        \sum_{i=1}^{n-1} t_i \exp(\alpha_0 T_{i}) 
        \frac{\exp(\alpha_0 t_i) - \exp(\alpha t_i)}{\exp(\alpha t_i) - 1} ,
    \end{equation}
    where $T_i = \sum_{j=1}^{i-1} t_j$.
\end{theorem}

Theorem \ref{theorem:convergence of estimator equation} also implies the consistency as the approximate estimator $\hat{\alpha}$:

\begin{corollary}[Convergence of the approximate estimator]

As $X_{\min} \rightarrow \infty$, $\hat{\alpha}\rightarrow\alpha_0$.
\end{corollary}
This is because $g^*(\alpha)=0$ is monotonous and has only one root $\alpha=\alpha_0$.

\subsubsection{Asymptotic Monotonicity of the Estimation Equation}

We now show that $h(\alpha)$ in \eqref{eq:approximate estimator equation} becomes monotonously decreasing in $\alpha$ in the limit, which is beneficial for optimization.

\begin{theorem}[Asymptotic Monotonicity of $h(\alpha)$]
\label{theorem:Asymptotic Monotonicity of h(alpha)}
    The function $h(\alpha)$ is monotonously decreasing in $\alpha$ as $X_{\min} \rightarrow \infty$ when $\alpha>0$.
\end{theorem}

This is obvious since $P(X_{i+1}> X_i)=0$ as $X_i \rightarrow \infty$ under $\alpha>0$, implying the signs of $\frac{1}{\exp(\alpha t_i)-1}$ in $h(\alpha)$ are positive.

When $X_i$ are finite, $h(\alpha)$ may not be strictly decreasing—especially when the sequence contains both growth and decline (i.e., $X_{i+1}<X_i$ and $X_{j+1}>X_j$), but it remains significantly easier to solve $h(\alpha)=0$ than maximizing the original log-likelihood $l$ in most practical settings.

\subsubsection{Identity to MLE under Equidistant Timepoints}

Same as the Gaussian approximation, for equidistant timepoint data, the approximate MLE is equivalent to the Galton-Watson estimator in Section \ref{subsubsec:Galton-Watson} (which is also the MLE).

\begin{theorem}
\label{theorem:Approximate Estimator Equals MLE}
    When time intervals are equidistant, the approximate estimator is identical to the Galton-Watson estimator (MLE for the LBDP).
\end{theorem}

This can be verified by taking $t_i=\tau$ in the estimator equation \eqref{eq:approximate estimator equation}.

\subsubsection{Robustness to Zero Values}

The approximate estimator $\hat{\alpha}$ is robust to zero values in the sequence $\{X_i\}$, as the equation denominator in \eqref{eq:approximate estimator equation} does not directly involve $X_i$. This property enhances numerical stability in datasets where extinction or dropout events occur.

\subsubsection{Observation Weighting in the Estimation}

Finally, let's revisit the approximate estimator as the solution of equation \eqref{eq:approximate estimator equation}:

\[
\sum_{i=1}^{n-1} \frac{t_i}{\exp(\alpha t_i) - 1} \left[X_{i+1} - X_i \exp(\alpha t_i)\right] = 0
\]

The equation effectively illustrates how different observations are weighted during the estimation process. Specifically, each term in the summation is scaled by a weight of the form \(\frac{t_i}{\exp(\alpha t_i) - 1}\), which depends on both the time interval \(t_i\) and the current value of the parameter \(\alpha\). This weighting scheme naturally adjusts the influence of each observation pair \((X_i, X_{i+1})\), placing greater emphasis on time intervals where exponential growth or decay is more stable and less susceptible to numerical instability. Consequently, the estimator adapts to the temporal spacing of the data, modulating the contribution of each term based on both the timing and the dynamic behavior of the observed process. This results in a more robust and context-aware estimation of \(\alpha\). In essence, the estimation procedure can be further simplified to a form of regression with weighted squared errors:

\begin{equation}
    X(t_i) = X(0)\exp(\alpha t_i) +\sqrt{\frac{t_i}{\exp(\alpha t_i) - 1}} \cdot \epsilon_i
\end{equation}

where each data point is weighted according to \(\frac{t_i}{\exp(\alpha t_i) - 1}\). We would not expand on this but encourage further attempts to formalize the LBDP problem as a weighted least-squares regression.

\subsection{Extension to time-nonhomogeneous birth-death processes}

The approximate estimator can be generalized for a generalized birth-death process in \cite{kendall1948generalized}: 

\begin{equation}
\begin{aligned}
    n_{t+dt} &= n_t + 1,  & \text{with associated probability } \lambda(t)n_t\, dt + o(dt);\\
    n_{t+dt} &= n_t,      & \text{with associated probability } 1 - \{\lambda(t) + \mu(t)\}n_t\, dt + o(dt);\\
    n_{t+dt} &= n_t - 1,  & \text{with associated probability } \mu(t)n_t\, dt + o(dt)
\end{aligned}
\end{equation}

where the birth rate $\lambda(t)$ and death rate $\mu(t)$ can be functions of time $t$ instead of constants. 

As we do not define the functional form of $\lambda(t)$ and $\mu(t)$, we rewrite $\mu_X(t)=\mu_X(t;\theta)$ and $\sigma_X^2(t)=\sigma_X^2(t;\theta)$ where $\theta$ represents all parameters to be estimated in the generalized birth-death process. The Gaussian approximation still holds by the homogeneous nature of the branching processes:

\begin{align} \label{eq:Gaussian_approximation Generalized}
X(t) \vert X(0) \sim N\left( X(0) \mu_{X}(t;\theta), X(0) \sigma_X^2(t;\theta) \right),
\end{align}

where

\begin{align}
\begin{split}
    \label{eq:Gaussian_approximation Generalized moments}
    \mu_{X}(t;\theta) &=  e^{\int_0^t \lambda(s;\theta)-\mu(s;\theta) ds}, \\
    \sigma_X^2(t;\theta) &= \int_0^t (\lambda(u;\theta) + \mu(u;\theta)) \exp\left(-\int_0^u (\lambda(s;\theta) - \mu(s;\theta)) \, ds\right) du.
\end{split}
\end{align}

Suppose we have $n$ observations $X_1,\dots, X_n$ at $n$ timepoints $T_1,\dots,T_n$, forming $n-1$ time intervals $t_1,\dots,t_{n-1}$, we can rewrite the transition density as

\begin{align} \label{eq:transition Generalized}
f_{X_{i+1} \vert X_i}(x)  = \frac{1}{\sqrt{2\pi X_i \sigma_i^2(\theta)}} \exp\left(-\frac{\left[x-X_i \mu_i(\theta)\right]^2}{2 X_i \sigma_i^2(\theta)}\right)
\end{align}

where

\begin{align}
\begin{split}
    \label{eq:transition Generalized moments}
    \mu_i(\theta) &=  e^{\int_{T_i}^{T_{i+1}} \lambda(s;\theta)-\mu(s;\theta) ds}, \\
    \sigma_i^2(\theta) &= \int_{T_i}^{T_{i+1}} \left[\lambda(u;\theta) + \mu(u;\theta)\right] \exp\left(-\int_{T_i}^{u} \left[\lambda(s;\theta) - \mu(s;\theta)\right] \, ds\right) du.
\end{split}
\end{align}

the joint log-likelihood is

\begin{equation}
l(\theta) = \sum_{i=1}^{n-1} \log f_{X_{i+1} \vert X_i}(x_{i+1})
\end{equation}

The generalized approximate estimator $\hat{\theta}$ is the solution of 

\begin{equation}
\label{eq:Generalized Approximate Estimator}
\sum_{i=1}^{n-1} \frac{\frac{\partial \mu_i(\theta)}{\partial \theta}}{\sigma^2_i(\theta)} \left[X_{i+1}-X_i \mu_i(\theta)\right] = 0
\end{equation}

We can show that as $X_{min}\rightarrow\infty$, the true parameter $\theta_0$ is one of the roots of this equation (number of the roots depends on the form of $\lambda(t;\theta)$ and $\mu(t;\theta)$).

\begin{theorem}[Convergence of the generalized estimator equation]
\label{theorem:Convergence of g(theta) generalized} As $X_{\min} \rightarrow \infty$, define function $g(\theta)$ as:
    
    \begin{equation}
    \label{eq:generalized g}
    g(\theta) =\frac{1}{X_1\sum_{i=1}^{n-1}\left[ \frac{\partial \mu_i(\theta_0)}{\partial \theta} \frac{\mu_i(\theta_0)}{\sigma_i^2(\theta_0)}\prod_{j=1}^{i-1}\mu_j(\theta_0)\right]} \left(\sum_{i=1}^{n-1} \frac{\frac{\partial \mu_i(\theta)}{\partial \theta}}{\sigma^2_i(\theta)} \left[X_{i+1}-X_i \mu_i(t_i;\theta)\right]\right).
    \end{equation}
    
    converges to a limiting function $g^*(\theta)$ given by:
    \begin{equation}
    \label{eq:limit of generalized g}
        g^*(\theta)= \frac{1}{\sum_{i=1}^{n-1}\left[\frac{\frac{\partial \mu_i(\theta_0)}{\partial \theta}}{\sigma_i^2(\theta_0)}\cdot \prod_{j=1}^{i}\mu_j(\theta_0)\right]}\sum_{i=1}^{n-1} \frac{\frac{\partial \mu_i(\theta)}{\partial \theta}}{\sigma^2_i(\theta)} \cdot \left[\mu_i(\theta_0)-\mu_i(\theta)\right]\cdot\prod_{j=1}^{i-1}\mu_j(\theta_0) ,
    \end{equation}
    where $\mu_i(\theta_0)$ and $\sigma_i^2(\theta_0)$ are the mean and variance functions at the true parameters $\theta_0$.  As solving \eqref{eq:approximate estimator equation} is also equivalent to solving $g(\theta)$, the estimator is also the solution of $g(\theta)$.
\end{theorem}

It is easy to see that when $\theta=\theta_0$, $g^*(\theta)=0$ as $\mu_i(\theta_0)-\mu_i(\theta)=0$ is a multiplier. So the true parameter $\theta_0$ is one of the roots of the equation above.

\section{Simulation Study: Comparisons with Saddlepoint Approximation and Gaussian Approximation}
\label{sec:simulation}

In order to assess the performance of the proposed approximate MLE for LBDP, we designed an extensive set of simulations. These simulations model the LBDP using both the exact Gillespie algorithm and the tau-leaping approximation. For each setting, we compared the following three estimation techniques: Gaussian approximation, saddlepoint approximation, and the approximate MLE. The non-equidistant time intervals $\{t_i\}$ are simulated using Gamma distributions.

\subsection{Simulation Framework}

For each simulation, we draw $n_t = 10$ observed timepoints from a Gamma distribution (rate = 1) with varying shape parameters to control time spread. We then simulate $n$ independent LBDP trajectories up to the maximum observed time using either the Gillespie or Tau-Leaping algorithm. Observations are obtained by aligning simulated trajectories with the closest preceding timepoints. Each method estimates the growth rate $\alpha = \lambda - \mu$, which we use for comparison. This process is repeated $M = 1000$ times to compute the average Mean Absolute Error (MAE) and runtime.

We considered various combinations of the birth rate $\lambda$, death rate $\mu$, initial population size $X_0$, and number of trajectories $n$, covering both small-scale and large-scale scenarios. Tau-leaping was used in cases where $b$ and $d$ were large to reduce computational overhead.

\subsection{Results Summary}

Table~\ref{tab:mae-results} presents the averaged results over $M = 1000$ simulations for representative parameter configurations, including variations in $(\lambda,\mu)$, $X_0$, and $n$.

\begin{table}[!ht]
    \centering
    \small  
    \caption{Mean Absolute Error (MAE) of $\hat{\alpha}$ for different values of $X_0$ and $n$, based on 1000 simulations.}
    \label{tab:mae-results}
    \renewcommand{\arraystretch}{1.2}
    \begin{tabular}{lcccccc}
        \toprule
        \multicolumn{7}{c}{$\lambda = 0.2$, $\mu = 0.1$, $t_i\sim Gamma(1,1)$ (Gillespie)} \\
        \midrule
        \textbf{Method} 
            & \makecell{$X_0 = 10$\\$n = 1$}
            & \makecell{$X_0 = 10$\\$n = 5$}
            & \makecell{$X_0 = 10$\\$n = 10$}
            & \makecell{$X_0 = 100$\\$n = 1$}
            & \makecell{$X_0 = 100$\\$n = 5$}
            & \makecell{$X_0 = 100$\\$n = 10$} \\
        \midrule
        Gaussian & \textbf{0.0381} & 0.0182 & 0.0120 & \textbf{0.0122} & 0.0056 & 0.0038 \\
        Saddlepoint & 0.0409 & \textbf{0.0176} & \textbf{0.0117} & 0.0122 & 0.0056 & 0.0038 \\
        Approx MLE & 0.0397 & 0.0177 & 0.0117 & 0.0122 & \textbf{0.0056} & \textbf{0.0038} \\
        \midrule
        \multicolumn{7}{c}{$\lambda = 0.6$, $\mu = 0.4$, $t_i\sim$ Gamma(1,1) (Gillespie)} \\
        \midrule
        Gaussian & \textbf{0.0642} & 0.0261 & 0.0179 & 0.0182 & 0.0083 & 0.0056 \\
        Saddlepoint & 0.0688 & 0.0256 & 0.0174 & 0.0184 & 0.0082 & 0.0057 \\
        Approx MLE & 0.0671 & \textbf{0.0253} & \textbf{0.0174} & \textbf{0.0182} & \textbf{0.0082} & \textbf{0.0056} \\
        \midrule
        \multicolumn{7}{c}{$\lambda = 2$, $\mu = 1$, $t_i\sim$ Gamma(0.2,1) (Tau-leaping)} \\
        \midrule
        Gaussian & 0.3436 & 0.1568 & 0.1342 & 0.5123 & 0.5615 & 0.4111 \\
        Saddlepoint & 0.3452 & 0.1683 & 0.1110 & 0.1076 & 0.0478 & 0.0387 \\
        Approx MLE & \textbf{0.3429} & \textbf{0.1436} & \textbf{0.0926} & \textbf{0.1021} & \textbf{0.0454} & \textbf{0.0366} \\
        \midrule
        \multicolumn{7}{c}{$\lambda = 6$, $\mu = 4$, $t_i\sim$ Gamma(0.2,1) (Tau-leaping)} \\
        \midrule
        Gaussian & 0.8960 & 0.6157 & 0.5827 & 1.7509 & 1.4949 & 1.5620 \\
        Saddlepoint & 0.5293 & \textbf{0.2267} & 0.1718 & 0.1684 & 0.0979 & 0.0809 \\
        Approx MLE & \textbf{0.4985} & 0.2275 & \textbf{0.1573} & \textbf{0.1590} & \textbf{0.0940} & \textbf{0.0707} \\
        \bottomrule
    \end{tabular}
\end{table}

\begin{table}[!ht]
    \centering
    \small
    \caption{Average runtime (in seconds) of $\hat{\alpha}$ for different values of $X_0$ and $n$, based on 1000 simulations.}
    \label{tab:runtime-results}
    \renewcommand{\arraystretch}{1.2}
    \begin{tabular}{lcccccc}
        \toprule
        \multicolumn{7}{c}{$\lambda = 0.2$, $\mu = 0.1$, $t_i\sim Gamma(1,1)$ (Gillespie)} \\
        \midrule
        \textbf{Method} 
            & \makecell{$X_0 = 10$\\$n = 1$}
            & \makecell{$X_0 = 10$\\$n = 5$}
            & \makecell{$X_0 = 10$\\$n = 10$}
            & \makecell{$X_0 = 100$\\$n = 1$}
            & \makecell{$X_0 = 100$\\$n = 5$}
            & \makecell{$X_0 = 100$\\$n = 10$} \\
        \midrule
        Gaussian & 0.01266 & 0.04642 & 0.08805 & 0.01137 & 0.04810 & 0.09464 \\
        Saddlepoint & 0.02932 & 0.13009 & 0.25895 & 0.03478 & 0.16414 & 0.33271 \\
        Approx MLE & \textbf{0.00038} & \textbf{0.00076} & \textbf{0.00124} & \textbf{0.00031} & \textbf{0.00072} & \textbf{0.00123} \\
        \midrule
        \multicolumn{7}{c}{$\lambda = 0.6$, $\mu = 0.4$, $t_i\sim$ Gamma(1,1) (Gillespie)} \\
        \midrule
        Gaussian & 0.01321 & 0.05309 & 0.10189 & 0.01282 & 0.05438 & 0.10402 \\
        Saddlepoint & 0.02185 & 0.11033 & 0.21929 & 0.02904 & 0.15040 & 0.29932 \\
        Approx MLE & \textbf{0.00030} & \textbf{0.00071} & \textbf{0.00116} & \textbf{0.00030} & \textbf{0.00072} & \textbf{0.00119} \\
        \midrule
        \multicolumn{7}{c}{$\lambda = 2$, $\mu = 1$, $t_i\sim$ Gamma(0.2,1) (Tau-leaping)} \\
        \midrule
        Gaussian & 0.01779 & 0.08423 & 0.15843 & 0.02368 & 0.11862 & 0.24849 \\
        Saddlepoint & 0.02679 & 0.13290 & 0.24998 & 0.02994 & 0.14670 & 0.31725 \\
        Approx MLE & \textbf{0.00031} & \textbf{0.00071} & \textbf{0.00116} & \textbf{0.00031} & \textbf{0.00073} & \textbf{0.00128} \\
        \midrule
        \multicolumn{7}{c}{$\lambda = 6$, $\mu = 4$, $t_i\sim$ Gamma(0.2,1) (Tau-leaping)} \\
        \midrule
        Gaussian & 0.02224 & 0.10628 & 0.19534 & 0.02658 & 0.12712 & 0.24390 \\
        Saddlepoint & 0.03177 & 0.15338 & 0.29864 & 0.03425 & 0.16749 & 0.34210 \\
        Approx MLE & \textbf{0.00035} & \textbf{0.00073} & \textbf{0.00114} & \textbf{0.00035} & \textbf{0.00072} & \textbf{0.00117} \\
        \bottomrule
    \end{tabular}
\end{table}

The simulation results demonstrate a clear and consistent advantage of the proposed approximate MLE estimator across all experimental configurations. As shown in Table~\ref{tab:mae-results}, the approximate MLE achieves the lowest mean absolute error (MAE) in nearly all settings, with especially pronounced gains under the more challenging scenarios—such as large birth and death rates (e.g., $\lambda = 2$, $\mu = 1$ and $\lambda = 6$, $\mu = 4$) and when using tau-leaping with high variability in observation times. In these cases, both the Gaussian and saddlepoint methods suffer substantial accuracy degradation, while the approximate MLE remains remarkably stable and accurate.

Runtime comparisons in Table~\ref{tab:runtime-results} further reinforce the practical value of the approximate MLE. It is consistently the fastest method, with runtimes several orders of magnitude lower than those of the Gaussian and saddlepoint approaches, regardless of population size or number of trajectories. Notably, even as $X_0$ and $n$ increase, the runtime of the approximate MLE remains effectively constant and negligible, whereas the Gaussian and especially the saddlepoint methods scale poorly in computational cost.

Together, these results highlight the approximate MLE as a highly robust, accurate, and scalable estimator, particularly well-suited for real-time or large-scale inference tasks where both speed and precision are paramount.

\section{Real-World Data Application: Clonal Hematopoiesis}

We applied the approximate MLE method alongside the saddlepoint and Gaussian approximations to the SardiNIA dataset from \cite{fabre2022longitudinal}, a longitudinal study of clonal hematopoiesis of indeterminate potential (CHIP). The goal was to estimate the clonal growth rate $\alpha = \lambda - \mu$, often referred to as the ``fitness effect,'' which quantifies the selective advantage of mutant hematopoietic clones relative to wild-type cells.

The SardiNIA study followed participants from four towns in Sardinia over five phases of data collection spanning more than 20 years. We focused our analysis on four recurrent DNMT3A mutations: R736C, R736H, R882C, and R882H. These were chosen due to prior estimates of their fitness effects reported by \cite{fabre2022longitudinal} and \cite{watson2020evolutionary}, enabling direct methodological comparison.

\begin{figure}[!ht]
\centering
\includegraphics[width=0.9\textwidth]{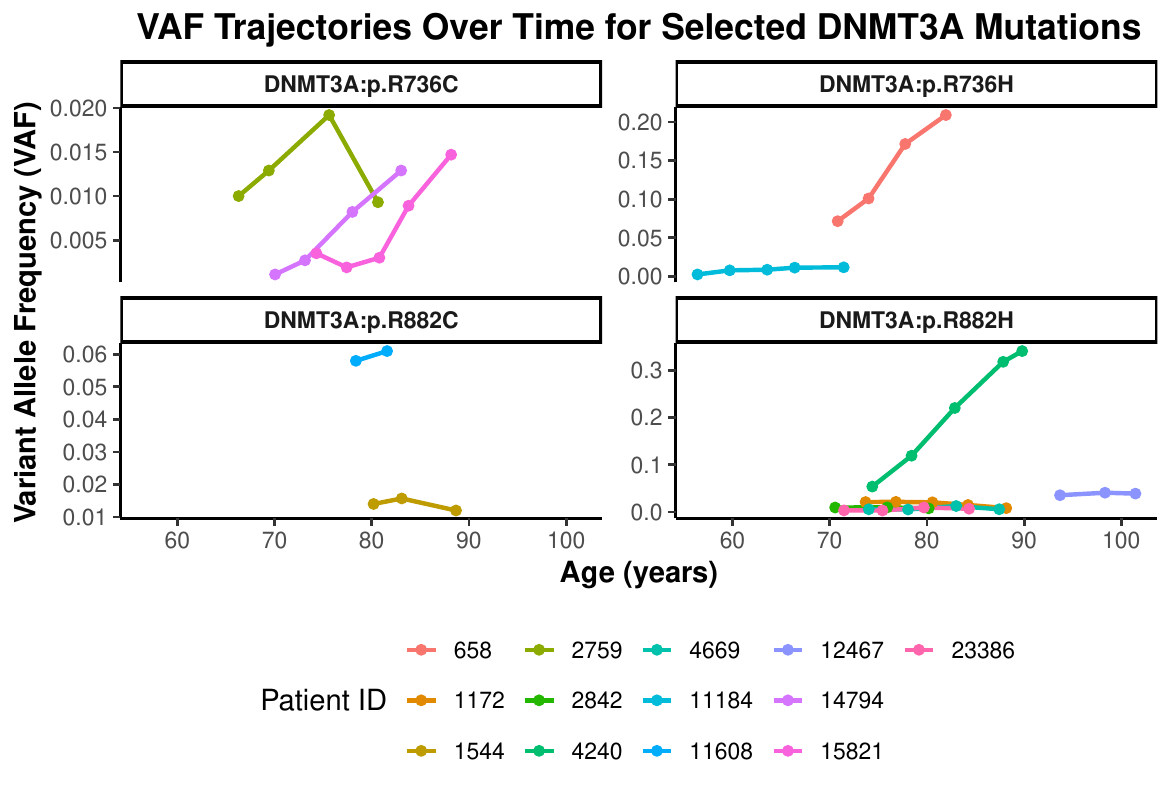}
\caption{VAF trajectories over time for selected \textit{DNMT3A} mutations.}
\label{fig:sard-vaf}
\end{figure}

Since the dataset provides variant allele frequencies (VAFs) rather than absolute counts, we used a transformation to estimate mutant cell counts. Assuming a constant population of 200{,}000 wild-type cells, the VAF is modeled as:
\[
\text{VAF}_{\text{mut}} = \frac{1}{2} \cdot \frac{X}{X + 200{,}000}
\]
Solving for $X$, the estimated count of mutant cells, gives:
\[
X = \frac{\text{VAF}_{\text{mut}}}{1 - \text{VAF}_{\text{mut}}} \times 2 \times 200{,}000
\]

This transformation enables the application of count-based models to VAF data.

$$
\text{VAF of Mutant} = \frac{1}{2}\times \frac{\text{Count of Mutant}}{\text{Count of Mutant} + \text{Count of Wildtype}} =\frac{1}{2}\times\frac{\text{Count of Mutant}}{\text{Count of Mutant}+200000}
$$

Our transformation is the inverse transformation:

$$
\text{Count of Mutant} = \frac{\text{VAF of Mutant}}{1-\text{VAF of Mutant}}\times 200000\times2
$$

\begin{table}[!ht]
    \centering
    \small
    \caption{Estimated growth rates ($\hat{\alpha}$, in \% per year) for selected \textit{DNMT3A} mutations using three methods. 95\% confidence intervals are based on the 2.5\% and 97.5\% quantiles.}
    \label{tab:dnmt3a-alpha}
    \renewcommand{\arraystretch}{1.2}
    \begin{tabular}{llcc}
        \toprule
        \textbf{Mutation} & \textbf{Method} & \textbf{Mean} & \textbf{95\% CI} \\
        \midrule
        DNMT3A R736C & Approx. MLE       & 9.16  & [0.35, 14.4] \\
                     & Gaussian Approx. & 9.40  & [1.03, 14.5] \\
                     & Saddlepoint      & 9.19  & [0.37, 14.4] \\
                     & Fabre (2022)     & 10.3  & [0.7, 14.1] \\
                     & Watson (2020)    & 12.3  & [11.6, 13.4] \\
        \midrule
        DNMT3A R736H & Approx. MLE       & 8.38  & [6.82, 9.94] \\
                     & Gaussian Approx. & 50.7  & [4.21, 97.1] \\
                     & Saddlepoint      & 8.39  & [6.83, 9.95] \\
                     & Fabre (2022)     & 10.3  & [3.2, 17.6] \\
                     & Watson (2020)    & 14.1  & [13.2, 15.4] \\
        \midrule
        DNMT3A R882C & Approx. MLE       & -0.006 & [-1.60, 1.59] \\
                     & Gaussian Approx. & 3.63   & [1.77, 5.50] \\
                     & Saddlepoint      & -0.005 & [-1.60, 1.59] \\
                     & Fabre (2022)     & 1.8    & [-6.8, 11.5] \\
                     & Watson (2020)    & 18.7   & [18.2, 19.4] \\
        \midrule
        DNMT3A R882H & Approx. MLE       & 1.64  & [-4.76, 10.5] \\
                     & Gaussian Approx. & 35.7  & [1.66, 99.7] \\
                     & Saddlepoint      & 1.65  & [-4.76, 10.5] \\
                     & Fabre (2022)     & 5.5   & [0.7, 10.9] \\
                     & Watson (2020)    & 14.8  & [14.1, 15.7] \\
        \bottomrule
    \end{tabular}
\end{table}

Table~\ref{tab:dnmt3a-alpha} presents the estimated growth rates across methods. For most mutations, the approximate MLE and saddlepoint methods yield estimates consistent with those reported by \cite{fabre2022longitudinal}, indicating reliability. However, the Gaussian approximation produced substantially higher and more variable estimates for R736H and R882H, highlighting its sensitivity to noisy or sparse VAF measurements. This sensitivity may stem from the Gaussian model's dependence on the absolute scale of count data, which can be unstable when derived from compositional VAF values.

In contrast, both the approximate MLE and classical Galton–Watson estimators are invariant to the arbitrary scaling factor used to convert VAF to counts. This makes them particularly suitable for compositional datasets like CHIP, where absolute cell counts are unknown. While we did not directly assess the saddlepoint approximation's sensitivity to scaling in this analysis, its formulation also depends on absolute count values, suggesting potential vulnerability to the same issue. Accordingly, the approximate MLE emerges as a stable and computationally efficient method for estimating clonal fitness from real-world longitudinal sequencing data.

\section{Discussion}

We introduced a fast and accurate approximate MLE for LBDP based on the Gaussian approximation to simplify inference to a one-dimensional optimization problem. This leads to substantial gains in runtime while maintaining high accuracy across simulated and real-world scenarios.

Simulations show that our method matches the estimation accuracy of saddlepoint and Gaussian approaches, while being significantly faster, especially under large sample sizes. Unlike the Gaussian approximation, which is highly sensitive to scaling and data sparsity, and the saddlepoint method, which depends on absolute count information, our approach is invariant to scaling transformations. This makes it well-suited for compositional data where true counts are unknown, such as variant allele frequencies (VAFs).

By applying the method to the SardiNIA CHIP dataset \citep{fabre2022longitudinal}, we estimated the clonal growth rate $\alpha = \lambda - \mu$ for key DNMT3A mutations. The approximate MLE and saddlepoint estimates aligned well with prior studies, while the Gaussian approximation produced unstable results for mutations like R736H and R882H, highlighting its sensitivity in noisy, compositional datasets.

The approximate MLE offers a fast, stable, and interpretable approach for estimating LBDP parameters from longitudinal data. It performs well under realistic conditions, including non-equidistant observations and compositional data with observational noise, and is particularly useful for modeling clonal dynamics. While not explored empirically here, the method also extends naturally to time-nonhomogeneous birth–death models, suggesting further potential applications.

\bibliographystyle{plain}
\bibliography{sample}

\begin{thebibliography}{10}

\bibitem{biesiadny2022statistical}
Sara Biesiadny.
\newblock {\em Statistical Analysis of Clonal Hematopoiesis: A Stochastic Modeling Approach}.
\newblock PhD thesis, Rice University, 2022.

\bibitem{crawford2018computational}
Forrest~W Crawford, Lam Si~Tung Ho, and Marc~A Suchard.
\newblock Computational methods for birth-death processes.
\newblock {\em Wiley Interdisciplinary Reviews: Computational Statistics}, 10(2):e1423, 2018.

\bibitem{davison2021parameter}
Anthony~C Davison, Sophie Hautphenne, and Andrea Kraus.
\newblock Parameter estimation for discretely observed linear birth-and-death processes.
\newblock {\em Biometrics}, 77(1):186--196, 2021.

\bibitem{fabre2022longitudinal}
Margarete~A Fabre, Jos{\'e}~Guilherme de~Almeida, Edoardo Fiorillo, Emily Mitchell, Aristi Damaskou, Justyna Rak, Valeria Orr{\`u}, Michele Marongiu, Michael~Spencer Chapman, MS~Vijayabaskar, et~al.
\newblock The longitudinal dynamics and natural history of clonal haematopoiesis.
\newblock {\em Nature}, 606(7913):335--342, 2022.

\bibitem{guttorp1991statistical}
Peter Guttorp.
\newblock {\em Statistical inference for branching processes}.
\newblock John Wiley \& Sons, 1991.

\bibitem{immel1951problems}
Eric~Robert Immel.
\newblock {\em Problems of estimation and of hypothesis testing connected with birth-and-death Markov processes}.
\newblock PhD thesis, University of California, Los Angeles--Mathematics, 1951.

\bibitem{keiding1975maximum}
Niels Keiding.
\newblock Maximum likelihood estimation in the birth-and-death process.
\newblock {\em The Annals of Statistics}, 3(2):363--372, 1975.

\bibitem{kendall1948generalized}
David~G Kendall.
\newblock On the generalized" birth-and-death" process.
\newblock {\em The annals of mathematical statistics}, 19(1):1--15, 1948.

\bibitem{louca2020extant}
Stilianos Louca and Matthew~W Pennell.
\newblock Extant timetrees are consistent with a myriad of diversification histories.
\newblock {\em Nature}, 580(7804):502--505, 2020.

\bibitem{tavare2018linear}
Simon Tavar{\'e}.
\newblock The linear birth--death process: an inferential retrospective.
\newblock {\em Advances in Applied Probability}, 50(A):253--269, 2018.

\bibitem{watson2020evolutionary}
Caroline~J Watson, AL~Papula, Gladys~YP Poon, Wing~H Wong, Andrew~L Young, Todd~E Druley, Daniel~S Fisher, and Jamie~R Blundell.
\newblock The evolutionary dynamics and fitness landscape of clonal hematopoiesis.
\newblock {\em Science}, 367(6485):1449--1454, 2020.

\bibitem{wu2023using}
C~Wu, EB~Gunnarsson, EM~Myklebust, A~K{\"o}hn-Luque, DS~Tadele, JM~Enserink, A~Frigessi, J~Foo, and K~Leder.
\newblock Using birth-death processes to infer tumor subpopulation structure from live-cell imaging drug screening data.
\newblock {\em arXiv preprint arXiv:2303.08245}, 2023.

\bibitem{zarebski2022computationally}
Alexander~Eugene Zarebski, Louis du~Plessis, Kris~Varun Parag, and Oliver~George Pybus.
\newblock A computationally tractable birth-death model that combines phylogenetic and epidemiological data.
\newblock {\em PLOS Computational Biology}, 18(2):e1009805, 2022.

\end{thebibliography}

\newpage

\section{Appendix}

\subsection{Proof for Theorem \ref{theorem:Dominating term}}
In Theorem \ref{theorem:Dominating term}, we have three terms:
\label{subsec:dominating term}

\begin{align}
\label{eq:likelihood decomposition 2}
\frac{\partial l}{\partial \alpha} &= \frac{t_i\Bigl[X_{i+1} - X_i\exp(\alpha t_i)\Bigr]}{\sigma^2\Bigl[\exp(\alpha t_i)-1\Bigr]}  \\
&\quad + \frac{t_i\Bigl[X_{i+1} - X_i\exp(\alpha t_i)\Bigr]^2\Bigl[2\exp(\alpha t_i)-1\Bigr]}{2X_i\sigma^2\exp(\alpha t_i)\Bigl[\exp(\alpha t_i)-1\Bigr]^2}  \\
&\quad -\frac{t_i}{2} - \frac{t_i\exp(\alpha t_i)}{2[\exp(\alpha t_i)-1]}.
\end{align}

Let $\frac{\partial l}{\partial \alpha }=l'$, the first term be $l_1$, the second term be $l_2$ and the third term be $l_3$, we have

\begin{align}
\begin{split}
& l_1 = \frac{t_i\Bigl[X_{i+1} - X_i\exp(\alpha t_i)\Bigr]}{\sigma^2\Bigl[\exp(\alpha t_i)-1\Bigr]} \\
& l_2 = \frac{t_i\Bigl[X_{i+1} - X_i\exp(\alpha t_i)\Bigr]^2\Bigl[2\exp(\alpha t_i)-1\Bigr]}{2X_i\sigma^2\exp(\alpha t_i)\Bigl[\exp(\alpha t_i)-1\Bigr]^2}\\
& l_3 = -\frac{t_i}{2} - \frac{t_i\exp(\alpha t_i)}{2[\exp(\alpha t_i)-1]}\\
\end{split}
\end{align}

where $\frac{\partial l}{\partial \alpha}=l' = l_1 + l_2 + l_3$.

Our target is to prove that $l_1$ is the dominating term by proving 

\begin{align}
&\frac{l_2}{l_1}\xrightarrow{p}0 \\
&\frac{l_3}{l_1}\xrightarrow{p}0 
\end{align}

as $X_i\rightarrow\infty$.

For $\frac{l1}{l2}$ we have

\begin{align}
\frac{l_1}{l_2} &= \frac{X_i\exp(\alpha t_i)}{X_{i+1} - X_i\exp(\alpha t_i)}\cdot \frac{2\exp(\alpha t_i)-2}{2\exp(\alpha t_i)-1} \\
& = \frac{X_i\exp(\alpha t_i)}{X_{i+1} - X_i\exp(\alpha t_i)}\cdot (\frac{1}{1+\frac{1}{2\exp(\alpha t_i)-2}})\\
& = \frac{1}{Z_i}
\end{align}

where 

\begin{equation}
\label{eq:Zi}
    Z_i\sim N(0,\frac{\sigma^2[2\exp(\alpha t_i)-1]^2}{2X_i \exp(\alpha t_i)(2\exp(\alpha t_i)-2)}
\end{equation}

By $Var(Z_i)\in \mathcal{O}(\frac{1}{X_i})$, we know that as $X_i\rightarrow\infty$, $\vert Z_i \vert<  X_i^{-\frac{1}{2}+\epsilon}$, $\forall \epsilon>0$ with probability 1, which means that 

\begin{align}
P(\vert \frac{l_1}{l_2}\vert > X_i^{\frac{1}{2}-\epsilon} )\rightarrow1 \quad \text{as } X_i \rightarrow \infty
\end{align}

for all $\epsilon>0$, or equivalently (by letting $\epsilon<\frac{1}{2}$, as $X_i^{\frac{1}{2}-\epsilon}\rightarrow \infty$),

\begin{equation}
\label{eq:l2/l1}
    \frac{l_2}{l_1}\xrightarrow{p} 0 \quad \text{as } X_i \rightarrow \infty
\end{equation}

For $\frac{l_1}{l_3}$, we have

\begin{equation}
\frac{l_1}{l_3} \propto \frac{\Bigl[X_{i+1} - X_i\exp(\alpha t_i)\Bigr]}{\sigma^2\Bigl[2\exp(\alpha t_i)-1\Bigr]} = \tilde{Z_i}
\end{equation}

where

\begin{equation}
\label{eq:tilde{Zi}}
\tilde{Z_i}\sim N(0,X_i\frac{\exp(\alpha t_i)\left[\exp(\alpha t_i)-1\right] }{\sigma^2\Bigl[2\exp(\alpha t_i)-1\Bigr]^2})
\end{equation}

Similar to $\frac{l_1}{l_2}$, as $X_i\rightarrow\infty$, $\vert \tilde{Z_i} \vert<  X_i^{\frac{1}{2}-\epsilon}$, $\forall \epsilon>0$, with probability 1, which means that

\begin{align}
P(\vert \frac{l_1}{l_3}\vert > X_i^{\frac{1}{2}-\epsilon} )\rightarrow1 \quad \text{as } X_i \rightarrow \infty
\end{align}

or equivalently (by letting $\epsilon<\frac{1}{2}$, as $X_i^{\frac{1}{2}-\epsilon}\rightarrow \infty$),

\begin{equation}
\label{eq:l3/l1}
    \frac{l_3}{l_1}\xrightarrow{p} 0 \quad \text{as } X_i \rightarrow \infty
\end{equation}

By \ref{eq:l2/l1} and \ref{eq:l3/l1}, as $l'=l_1+l_2+l_3$, we have

\begin{align}
\label{eq:l'/l1}
P(\vert \frac{l'-l_1}{l_1}\vert<X_i^{-\frac{1}{2}+\epsilon} )\rightarrow1, \ \forall\epsilon>0\quad \text{as } X_i \rightarrow \infty
\end{align}

which means (by letting $\epsilon<\frac{1}{2}$, as $X_i^{\frac{1}{2}-\epsilon}\rightarrow \infty$),

\begin{equation}
    \frac{l'}{l_1}\xrightarrow{p} 1,  \quad \text{as } X_i \rightarrow \infty
\end{equation}

Finally, we show that $l_1$ is the dominating term in $\frac{\partial l}{\partial \alpha}$ with  $\frac{\frac{\partial l}{\partial \alpha}}{l_1}\xrightarrow{p} 1 $ as $X_i \rightarrow \infty$.

\subsection{Proof for Theorem \ref{theorem:convergence of estimator equation}}

\begin{align}
&g(\alpha) =\frac{1}{X_1\sum_{i=1}^{n-1} t_i\exp(\alpha_0 T_i)} \left[\sum_{i=1}^{n-1}  \frac{t_i}{\exp(\alpha t_i)-1}(X_{i+1}-X_i) - \sum_{i=1}^{n-1} t_i X_i\right] = 0,\\
&g^*(\alpha) = \frac{1}{X_1\sum_{i=1}^n t_i\exp(\alpha_0 T_i)} \sum_{i=1}^n t_i\exp(\alpha_0 T_{i}) \frac{\exp(\alpha_0 t_i)-1}{\exp(\alpha t_i)-1} -1
\end{align}

Our target is to show that as $X_{\text{min}}=\min(X_1,\dots,X_n)\rightarrow \infty$, we have $g\rightarrow g^*$.

To prove the convergence of $g(\alpha)$, we decomposite $g(\alpha)$ to two terms:

\begin{align}
\label{eq:the first term}
g(\alpha) &=\frac{1}{X_1\sum_{i=1}^{n-1} t_i\exp(\alpha_0 T_i)} \sum_{i=1}^{n-1}  \frac{t_i}{\exp(\alpha t_i)-1}(X_{i+1}-X_i) - \\
\label{eq:the second term}
&\frac{1}{X_1\sum_{i=1}^{n-1} t_i\exp(\alpha_0 T_i)}\sum_{i=1}^{n-1} t_i X_i.
\end{align}

Before proving the convergence of the two terms, we first show a useful result:

\begin{theorem}
    \label{theorem:convergence of Xn}
    \begin{equation}
    \lim_{X_{1}\rightarrow\infty}\frac{X_n}{X_1\exp(\alpha_0 T_n)} = 1  
    \end{equation}
\end{theorem}

\begin{proof}
    This can be easily shown using Slutsky's theorem or Strong Law of Large Numbers as $X_n$ is Gaussian distributed with mean $\mu_{X_n}=X_1\exp(\alpha_0 T_n)$ and variance $\sigma^2_{X_n}\rightarrow 0$.
\end{proof}

Then another little theorem:

\begin{theorem}
\label{theorem:convergence of sum of Xn}

For a series of random variable $Y_{i}(x)$ and functions $a_i(x)$ all continuous wrt $x$, if $\lim_{x\rightarrow\infty} \frac{Y_{i}(x)}{a_{i}(x)} = 1$ and $a_{i}(x)>0$ for $i=1,\dots,n$, then 

\begin{equation}
\lim_{x\rightarrow\infty} \frac{\sum_{i=1}^n Y_{i}(x)}{\sum_{i=1}^n a_i(x)} = 1 
\end{equation}

for all $n>0$.
\end{theorem}
\begin{proof}
By $\lim_{x\rightarrow\infty} \frac{Y_{i}(x)}{a_{i}(x)} = 1$, we have 

\begin{equation}
\lim_{x\rightarrow\infty}P(Y_i(x)>0)=1
\end{equation}

for all $i=1,\dots,n$.

By mediant inequality

\begin{equation}
\min\left(\frac{Y_1(x)}{a_1(x)},\dots\frac{Y_n(x)}{a_n(x)}\right)\leq \frac{\sum_{i=1}^n Y_i(x)}{\sum_{i=1}^n a_i(x)} \leq \max\left(\frac{Y_1(x)}{a_1(x)},\dots,\frac{Y_n(x)}{a_n(x)} \right)
\end{equation}

we can now take the limit

\begin{equation}
1=\lim_{x\rightarrow\infty}\min\left(\frac{Y_1(x)}{a_1(x)},\dots\frac{Y_n(x)}{a_n(x)}\right)\leq \frac{\sum_{i=1}^n Y_i(x)}{\sum_{i=1}^n a_i(x)} \leq \lim_{x\rightarrow\infty}\max\left(\frac{Y_1(x)}{a_1(x)},\dots,\frac{Y_n(x)}{a_n(x)} \right)=1
\end{equation}

Therefore,

\begin{equation}
\lim_{x\rightarrow\infty} \frac{\sum_{i=1}^n Y_i(x)}{\sum_{i=1}^n a_i(x)}=1
\end{equation}

\end{proof}

\subsubsection{Convergence of the second term}

We can use this theorem to prove the consistency of RHS:
let $T_i = \sum_{k=1}^{i-1} t_k$, by Theorem \ref{theorem:convergence of Xn}, we have $\frac{X_n}{\exp(\alpha_0 T_n)} \rightarrow 1 $ as $X_{\text{min}}\rightarrow \infty$, so 

\begin{align}
\eqref{eq:the second term} = \frac{1}{X_1\sum_{i=1}^{n-1} t_i\exp(\alpha_0 T_i)}\sum_{i=1}^{n-1} t_i X_i \rightarrow 1 \quad \text{a.s.}
\end{align}

\subsubsection{Convergence of the first term}

When doing estimation, we can only adjust the $\alpha$ in the denominator (true $\alpha$ is fixed and unknown), we use $\hat{\alpha}$ to imply that

\begin{align}
\begin{split}
\eqref{eq:the first term} &= \frac{1}{X_1\sum_{i=1}^{n-1} t_i\exp(\alpha_0 T_i)} \sum_{i=1}^{n-1}  \frac{t_i}{\exp(\alpha t_i)-1}(X_{i+1}-X_i) \\
& = \frac{1}{X_1\sum_{i=1}^{n-1} t_i\exp(\alpha_0 T_i)}  \sum_{i=1}^{n-1}  \frac{t_i}{\exp(\alpha t_i)-1}X_{i+1} - \frac{1}{\sum_{i=1}^{n-1} t_i\exp(\alpha_0 T_i)} \sum_{i=1}^{n-1}  \frac{t_i}{\exp(\alpha t_i)-1}X_i \\
& = \frac{1}{\sum_{i=1}^{n-1} t_i\exp(\alpha_0 T_i)}  \sum_{i=1}^{n-1}  \left[\frac{t_i\exp(\alpha_0 T_{i+1})}{\exp(\alpha t_i)-1}\frac{X_{i+1}}{X_1\exp(\alpha_0 T_{i+1})} - \frac{t_i\exp(\alpha_0 T_{i})}{\exp(\alpha t_i)-1}\frac{X_i}{X_1\exp(\alpha_0 T_{i})}\right] \\
& \rightarrow \frac{1}{\sum_{i=1}^{n-1} t_i\exp(\alpha_0 T_i)} \left(\sum_{i=1}^{n-1} \frac{t_i\exp(\alpha_0 T_{i+1})}{\exp(\alpha t_i)-1}-\frac{t_i\exp(\alpha_0 T_{i})}{\exp(\alpha t_i)-1}\right) \\
& = \frac{1}{\sum_{i=1}^{n-1} t_i\exp(\alpha_0 T_i)} \sum_{i=1}^{n-1} t_i\exp(\alpha_0 T_{i}) \frac{\exp(\alpha_0 t_i)-1}{\exp(\alpha t_i)-1}
\end{split}
\end{align}

\subsubsection{Convergence of \texorpdfstring{$g(\alpha)$}{g(alpha)}}
Finally, $g(\alpha)$ converges: 

$$
\lim_{X_{\text{min}}\rightarrow \infty}g(\alpha) = \frac{1}{\sum_{i=1}^n t_i\exp(\alpha_0 T_i)} \sum_{i=1}^n t_i\exp(\alpha_0 T_{i}) \frac{\exp(\alpha_0 t_i)-1}{\exp(\alpha t_i)-1} -1 = g^*(\alpha)
$$

and $g^*(\alpha)=0$ is equivalent to:

$$
\sum_{i=1}^n \frac{\exp(\alpha_0 t_i)-1}{\exp(\alpha t_i)-1} \cdot t_i\exp(\alpha_0 T_{i}) = \sum_{i=1}^n t_i\exp(\alpha_0 T_i)
$$

Because $t_i$ are greater than 0, we can see that the left hand side is monotonously decreasing wrt to $\alpha$, there is only one real root for the equation and it is $\alpha = \alpha_0$.

\subsection{Asymptotic Behavior of the Pseudo-Log-Likelihood}

We now define a pseudo-log-likelihood function $\tilde{l}$ derived from the approximate estimating equation and show its convergence to the true log-likelihood.

\begin{theorem}
\label{theorem:Asymptotic Convergence of pseudo-log-likelihood function}
Let $\tilde{l}(\alpha,\sigma^2)$ be the pseudo-log-likelihood defined by
\begin{align}
\label{eq:pseudo-log-likelihood}
    l(\alpha,\sigma^2) &= \tilde{l}(\alpha_0,\sigma^2) + \frac{1}{\sigma^2}
    \int_{\alpha_0}^\alpha \sum_{i=1}^{n-1} \frac{t_i}{\exp(a t_i) - 1} 
    \left[ X_{i+1} - X_i \exp(a t_i) \right] \, da \\
    &= l(\alpha_0,\sigma^2) + \frac{1}{\sigma^2}\sum_{i=1}^{n-1} (X_{i+1}-X_i)\left\vert\log\frac{\exp(\alpha_0 t_i)-1}{\exp(\alpha t_i)-1}\right\vert - X_{i+1} \vert \alpha_0-\alpha\vert t_i,
\end{align}

Then, as $X_{\min} \rightarrow \infty$, the relative error between $\tilde{l}(\alpha,\sigma^2)$ and the true log-likelihood $l(\alpha,\sigma^2)$ satisfies:
\begin{equation}
\label{eq:limit of tilde(l)}
    \left| \frac{\tilde{l}(\alpha,\sigma^2) - l(\alpha,\sigma^2)}{l(\alpha,\sigma^2)} \right| \in \mathcal{O}\left(\frac{1}{\log X_{\min}}\right),
\end{equation}
and hence
\begin{equation}
\label{eq:limit of tilde(l) 2}
    \frac{\tilde{l}(\alpha,\sigma^2)}{l(\alpha,\sigma^2)} \xrightarrow{P} 1.
\end{equation}
\end{theorem}

This result confirms that $\tilde{l}(\alpha,\sigma^2)$ is a reliable approximation of the true log-likelihood in the asymptotic regime.

\subsubsection{Proof for Theorem \ref{theorem:Asymptotic Convergence of pseudo-log-likelihood function}}

From \ref{subsec:dominating term}, we know that $l_1$ in \ref{subsec:dominating term} dominates $\frac{\partial l}{\partial\alpha}$ as $X_i\rightarrow\infty$ with $\frac{\frac{\partial l}{\partial \alpha}}{l_1}\xrightarrow{p} 1 $. More specifically, by \ref{eq:l2/l1} and \ref{eq:l3/l1}, we have

\begin{align}
    \frac{\frac{\partial l}{\partial \alpha}}{l_1} &= 1+ \frac{l_2}{l_1} + \frac{l_3}{l_1} \\
    & =1+  \frac{2\exp(\alpha t_i)-1}{2\exp(\alpha t_i)-2}\cdot \frac{X_{i+1} - X_i\exp(\alpha t_i)}{X_i\exp(\alpha t_i)}+\frac{\sigma^2\Bigl[2\exp(\alpha t_i)-1\Bigr]}{X_{i+1} - X_i\exp(\alpha t_i)} \\
    & = 1+ X_i^{-\frac{1}{2}}\cdot \frac{Z_i}{X_i^{-\frac{1}{2}}} + X_i^{-\frac{1}{2}}\cdot \frac{\tilde{Z_i}^{-1}}{X_i^{-\frac{1}{2}}} \\
    & = 1+ X_i^{-\frac{1}{2}}\cdot \left(\frac{Z_i}{X_i^{-\frac{1}{2}}} + \frac{\tilde{Z_i}^{-1}}{X_i^{-\frac{1}{2}}}\right)
\end{align}

by \ref{eq:Zi} and \ref{eq:tilde{Zi}}, we know that $\left(\frac{Z_i}{X_i^{-\frac{1}{2}}} + \frac{\tilde{Z_i}^{-1}}{X_i^{-\frac{1}{2}}}\right)$ is a random variable independent of $X_i$, therefore,

\begin{equation}
\label{eq:dl/da / l}
\frac{\frac{\partial l}{\partial \alpha}-l_1}{l_1} \in \mathcal{O}(X_i^{-\frac{1}{2}})
\end{equation}

and equivalently,

\begin{equation}
\label{eq:dl/da / l _ 2}
\frac{\frac{\partial l}{\partial \alpha}-l_1}{\frac{\partial l}{\partial \alpha}} \in \mathcal{O}(X_i^{-\frac{1}{2}})
\end{equation}

\subsubsection{Construction of pseudo-log-likelihood}
Suppose $X_{\text{min}}=\min(X_1,\dots,X_n)$. Define the pseudo-log-likelihood based on the left-hand side function in \ref{eq:approximate estimator equation}, and $\alpha_0$ is the true parameter in \ref{eq:Gaussian Approximation Likelihood}: 

\begin{align}
    \tilde{l}(\alpha,\sigma^2) &= l(\alpha_0,\sigma^2) + \frac{1}{\sigma^2}
    \int_{\alpha_0}^\alpha \sum_{i=1}^{n-1} \frac{t_i}{\exp(a t_i) - 1} 
    \left[ X_{i+1} - X_i \exp(a t_i) \right] \, da \\
     &= l(\alpha_0,\sigma^2) + \frac{1}{\sigma^2}\sum_{i=1}^{n-1} (X_{i+1}-X_i)\left\vert\log\frac{\exp(\alpha_0 t_i)-1}{\exp(\alpha t_i)-1}\right\vert - X_{i+1} \vert \alpha_0-\alpha\vert t_i,
\end{align}  
where $\tilde{l}(\alpha_0,\sigma^2) = l(\alpha_0,\sigma^2)$. When $X_{\text{min}}\rightarrow \infty$, by \eqref{eq:dl/da / l}, we have

\begin{equation}
\vert l_{1}-\frac{\partial l}{\partial \alpha}\vert  \in \mathcal{O}(X_i^{-\frac{1}{2}}l_{1})
\end{equation}

By the definition of $l_1$:

\begin{equation}
l_1 = \frac{t_i\Bigl[X_{i+1} - X_i\exp(\alpha t_i)\Bigr]}{\sigma^2\Bigl[\exp(\alpha t_i)-1\Bigr]} = \sqrt{\frac{X_i \exp(\alpha t_i)}{\sigma^2 \Bigl[\exp(\alpha t_i)-1\Bigr]}}\cdot Z_0
\end{equation}

where $Z_0\sim N(0,1)$ is a standard normal variable. This implies that

\begin{equation}
\label{eq:scale of l1}
l_1\in \mathcal{O}(X_{i}^{\frac{1}{2}})
\end{equation}

Thus,

\begin{equation}
\label{eq:scale of l1-dl/da}
\vert l_{1}-\frac{\partial l}{\partial \alpha}\vert  \in \mathcal{O}(1)
\end{equation}

Note that both $l_1$ and $l$ here only stands for $X_{i+1}\vert X_i$. From now we will use $l_{1,X_{i+1}\vert X_i}$
to denote $l_1$ in \eqref{eq:dl/da / l} for $X_{i+1}\vert X_i$ and $\frac{\partial l_{X_{i+1}\vert X_i}}{\partial \alpha}$ to denote $\frac{\partial l}{\partial \alpha}$ for $X_{i+1}\vert X_i$ in \eqref{eq:dl/da / l}. The new $l$ is the joint log-likelihood for $X_1,\dots,X_n$.

By summing up from $X_2\vert X_1$ to $X_n\vert X_{n-1}$, we have 

\begin{align}
\sum_{i=1}^{n-1}\vert l_{1,X_{i+1}\vert X_i}-\frac{\partial l_{X_{i+1}\vert X_i}}{\partial \alpha}\vert  \in \mathcal{O}(1)
\end{align}

by 

\begin{equation}
\sum_{i=1}^{n-1}\vert l_{1,X_{i+1}\vert X_i}-\frac{\partial l_{X_{i+1}\vert X_i}}{\partial \alpha}\vert\geq \vert \sum_{i=1}^{n-1}l_{1,X_{i+1}\vert X_i}-\sum_{i=1}^{n-1}\frac{\partial l_{X_{i+1}\vert X_i}}{\partial \alpha}\vert
\end{equation}

Note that by definition,

\begin{align}
&\sum_{i=1}^{n-1} l_{1,X_{i+1}\vert X_i}=\frac{\partial \tilde{l}}{\partial \alpha} \\ 
&\sum_{i=1}^{n-1}\frac{\partial l_{X_{i+1}\vert X_i}}{\partial \alpha}=\frac{\partial l}{\partial \alpha} 
\end{align}

where $\tilde{l}$ is the constructed pseudo-log-likelihood and $l$ is the true log-likelihood. Therefore, we have

\begin{equation}
\vert \frac{\partial \tilde{l}}{\partial \alpha}- \frac{\partial l}{\partial \alpha}\vert \in \mathcal{O}(1)\\
\end{equation}

As $\tilde{l}(\alpha_{0})=l(\alpha_{0})$ by definition, we have

\begin{equation}
\label{eq:scale of tilde(l)-l}
\tilde{l}(\alpha) - l(\alpha) = \int_{\alpha}^{\alpha_{0}}\frac{\partial \tilde{l}(a)}{\partial \alpha}- \frac{\partial l(a)}{\partial \alpha} da \in \mathcal{O}(1)
\end{equation}

wrt $X_{\text{min}}$. Revisit the definition of \eqref{eq:Gaussian Approximation Likelihood}, we have the log likelihood $l_{X_{i+1}\vert X_i}$ as a function of $X_i$: 

\begin{align}
l_{X_{i+1}\vert X_{i}} &\propto -\frac{1}{2}\log(X_i)-\frac{[X_{i+1} - X_i\exp(\alpha t_i)]^2 }{2X_i\exp(\alpha t_i)[\exp(\alpha t_i)-1]\sigma^2} \\
&\propto -\frac{1}{2}\log(X_i)- \frac{1}{2}\tilde{\chi}_1^2 \in \mathcal{O}(\log(X_i))
\end{align}

where $\tilde{\chi}_1^2$ is a chi-square random variable with 1 degree of freedom.

This means that for the joint likelihood:

\begin{align}
\label{eq:scale of l}
l\in \mathcal{O}(\sum_{i=1}^{n-1} \log(X_i))= O\left(\log(X_{\text{min}})\right)
\end{align}

Finally, by \eqref{eq:scale of tilde(l)-l}  and \eqref{eq:scale of l}, we have

\begin{equation}
\label{eq:limit of tilde(l)-Appendix}
    \left\vert \frac{\tilde{l}-l}{l} \right\vert  \in \mathcal{O}(\frac{1}{\log(X_{\text{min}})})
\end{equation}

Therefore, as $X_{\text{min}}\rightarrow \infty$, we have 
\begin{equation}
 \frac{\tilde{l}}{l} \xrightarrow{P} 1.
\end{equation}

\subsection{Proof for Theorem \ref{theorem:Convergence of g(theta) generalized}}

For generalized birth-death process with true parameter $\theta_0$, Theorem \ref{theorem:convergence of Xn} still holds, although we should rewrite it as

\begin{equation}
\lim_{X_{1}\rightarrow\infty}\frac{X_n}{E(X_n)} = 1  
\end{equation}

where by definition $E(X_n)=X_1 \prod_{j=1}^{i-1}\mu_j(\theta_0)$. 

Equivalently, let $X_{\text{min}}= \min(X_1,\dots,X_n)$,

\begin{equation}
\lim_{X_{\text{min}}\rightarrow\infty}\frac{X_n}{X_1} = \prod_{j=1}^{i-1}\mu_j(\theta_0). 
\end{equation}

We can decompose $g(\theta)$ in \eqref{eq:generalized g} as 

\begin{align}
\begin{split}
g(\theta) &=\frac{1}{X_1\sum_{i=1}^{n-1}\left[ \frac{\partial \mu_i(\theta_0)}{\partial \theta} \frac{\mu_i(\theta_0)}{\sigma_i^2(\theta_0)}\prod_{j=1}^{i-1}\mu_j(\theta_0)\right]} \left(\sum_{i=1}^{n-1} \frac{\frac{\partial \mu_i(\theta)}{\partial \theta} }{\sigma^2_i(\theta)} X_{i+1}-\sum_{i=1}^{n-1} \frac{\frac{\partial \mu_i(\theta)}{\partial \theta} \mu_i(\theta)}{\sigma^2_i(\theta)}X_i\right) \\
& = \frac{\sum_{i=1}^{n-1} \frac{\frac{\partial \mu_i(\theta)}{\partial \theta} }{\sigma^2_i(\theta)} X_{i+1}/X_1}{\sum_{i=1}^{n-1}\left[  \frac{\frac{\partial \mu_i(\theta_0)}{\partial \theta} }{\sigma_i^2(\theta_0)}\cdot \prod_{j=1}^{i}\mu_j(\theta_0)\right]} - \frac{\sum_{i=1}^{n-1} \frac{\frac{\partial \mu_i(\theta)}{\partial \theta}\mu_i(\theta)}{\sigma^2_i(\theta)} X_{i}/X_1}{\sum_{i=1}^{n-1}\left[ \frac{\partial \mu_i(\theta_0)}{\partial \theta} \frac{\mu_i(\theta_0)}{\sigma_i^2(\theta_0)}\cdot \prod_{j=1}^{i-1}\mu_j(\theta_0)\right]} 
\end{split}
\end{align}

where

\begin{align}
\begin{split}
    \mu_i(\theta) &=  e^{\int_{T_i}^{T_{i+1}} \lambda(s;\theta)-\mu(s;\theta) ds}, \\
    \sigma_i^2(\theta) &= \int_{T_i}^{T_{i+1}} \left[\lambda(u;\theta) + \mu(u;\theta)\right] \exp\left(-\int_{T_i}^{u} \left[\lambda(s;\theta) - \mu(s;\theta)\right] \, ds\right) du.
\end{split}
\end{align}

As $\lim_{X_{\text{min}}\rightarrow\infty}\frac{X_{i+1}}{X_1 } = \prod_{j=1}^{i}\mu_j(\theta_0) $ and $\lim_{X_{\text{min}}\rightarrow\infty}\frac{X_i}{X_1 } = \prod_{j=1}^{i-1}\mu_j(\theta_0) $, by Theorem \ref{theorem:convergence of sum of Xn}, we have

\begin{align}
\begin{split}
\lim_{X_{\min}\rightarrow\infty} g(\theta) &= \frac{\sum_{i=1}^{n-1} \frac{\frac{\partial \mu_i(\theta)}{\partial \theta}}{\sigma^2_i(\theta)} \prod_{j=1}^{i}\mu_j(\theta_0)}{\sum_{i=1}^{n-1}\left[  \frac{\frac{\partial \mu_i(\theta_0)}{\partial \theta}}{\sigma_i^2(\theta_0)}\cdot \prod_{j=1}^{i}\mu_j(\theta_0)\right]} - \frac{\sum_{i=1}^{n-1} \frac{\frac{\partial \mu_i(\theta)}{\partial \theta}\mu_i(\theta)}{\sigma^2_i(\theta)}\prod_{j=1}^{i-1}\mu_j(\theta_0)}{\sum_{i=1}^{n-1}\left[ \frac{\partial \mu_i(\theta_0)}{\partial \theta} \frac{\mu_i(\theta_0)}{\sigma_i^2(\theta_0)}\cdot \prod_{j=1}^{i-1}\mu_j(\theta_0)\right]} \\
& = \frac{1}{\sum_{i=1}^{n-1}\left[  \frac{\frac{\partial \mu_i(\theta_0)}{\partial \theta}}{\sigma_i^2(\theta_0)}\cdot \prod_{j=1}^{i}\mu_j(\theta_0)\right]}\sum_{i=1}^{n-1} \frac{\frac{\partial \mu_i(\theta)}{\partial \theta}}{\sigma^2_i(\theta)} \cdot \left[\mu_i(\theta_0)-\mu_i(\theta)\right]\cdot\prod_{j=1}^{i-1}\mu_j(\theta_0) 
\end{split}
\end{align}

When $\theta=\theta_0$, we have $\mu_i(\theta_0)-\mu_i(\theta)=0$, which will make $\lim_{X_{\min}\rightarrow\infty} g(\theta)=0$. By this, we have shown that the true parameter $\theta_0$ is one of the solutions of $\lim_{X_{\min}\rightarrow\infty} g(\theta)=0$.

\end{document}